\documentclass{IEEEtran}  

\usepackage{cite}
\usepackage{epsfig}

\usepackage{amsmath}
\usepackage{amsfonts}
\usepackage{amssymb}
\usepackage{amsthm}
\usepackage{bbm}
\usepackage{dsfont}

\newtheorem{assumption}{Assumption}

\newtheorem{theorem}{Theorem}
\newtheorem{remark}{Remark}
\newtheorem{lemma}{Lemma}
\newtheorem{corollary}{Corollary}

\newcommand{\ubar}[1]{\underline{#1}}
\newcommand{\bigO}{\mathcal{O}}

\begin{document}
\title{Asymptotic analysis of the Friedkin-Johnsen model when the matrix of the susceptibility weights approaches the identity matrix}

\author{Alfredo~Pironti
\thanks{A. Pironti is with the Department of Electrical Engineering and Information Technology, Universit\`a degli Studi di Napoli Federico II, Via Claudio 21, 80125 Napoli, Italy, e-mail: pironti@unina.it.}
}
\maketitle

\begin{abstract}
In this paper we analyze the Friedkin-Johnsen model of opinions when the coefficients weighting the agent susceptibilities to interpersonal influence approach $1$. We will show that in this case, under suitable assumptions, the model converges to a quasi-consensus condition among the agents. In general the achieved consensus value will be different from the one obtained by the corresponding DeGroot model.
\end{abstract}



\section{Introduction}
In this paper we consider the Friedkin-Johnsen (FJ) model of opinions~\cite{Friedkin}, which is described by the equation
\begin{equation} \label{eq:Friedkin1}
y(k+1)=\Lambda W y(k)+(I-\Lambda)y_0\,,\quad y(0)=y_0\,,
\end{equation}
where we will assume that $W\in \mathds{R}^{n\times n}$ is a nonnegative row stochastic irreducible matrix, and that $\Lambda$ is a diagonal matrix whose diagonal elements $\lambda_1, \lambda_2, \dots \lambda_n$ ranges in the interval $(0, 1)$.
In this model the components of the state vector $y$ represents the opinion of individuals (\emph{agents}) on a given subject (assumed in the interval $[0, 1]$), the elements $w_{ij}$ of the matrix $W$ represents the influence accorded by the agent $i$ to the agent $j$, and as explained in~\cite{Friedkin}, each coefficient $\lambda_i$ weights the \emph{agent susceptibilities to interpersonal influence}. 

For $\Lambda=I$ the FJ model reduces to the DeGroot model~\cite{DeGroot}
\begin{equation} \label{eq:DeGroot}
y(k+1)=W y(k)\,,\quad y(0)=y_0\,,
\end{equation}
which is know, under the additional assumption of $W$ being primitive, to converge to a consensus value, dictated by the left and right Perron-Frobenius eigenvectors of the matrix $W$. This means that for $k\rightarrow\infty$ the opinion of each agent converge to the same value. In general this does not happen in the FJ model, where the asymptotic value depends on the static gain matrix $H=(I-\Lambda W)^{-1}(I-\Lambda)$.

In this paper we will show that if all the coefficients $\lambda_i$ are sufficiently close to $1$, a quasi-consensus condition is reached, or in other words the agent opinions will converge to very close values. Note that, although the fact that for $\Lambda=I$ the FJ model coincides with the DeGroot model, this is not a trivial conclusion. Indeed, for one hand, the consensus value of the DeGroot model is a consequence of the asymptotic behavior of the zero input response of system~\eqref{eq:DeGroot}, whereas in our case the asymptotic solution of the FJ model will depend on the steady state response of system~\eqref{eq:Friedkin1}; on the other hand we will show that in general the quasi-consensus value achieved in the FJ is not only different from the one predicted by the DeGroot model, but that quasi-consensus is achieved also when the $W$ matrix is not primitive. Convergence to quasi-consensus, was already noted in~\cite{Prosku} for the particular case where $\Lambda=\alpha I$, i.e. when $\lambda_1=\lambda_2=\dots=\lambda_n=\alpha\rightarrow 1^-$; in this paper we consider the case in which $\Lambda$ approaches the identity matrix along a general direction. 

The FJ model plays an important role in the modeling of opinions in social networks, we refer to the early papers~\cite{DeGroot},~\cite{Friedkin}, and to the recent tutorial~\cite{Prosku} for an introduction to this subject. A general analysis of the behavior of the FJ model from a system theoretic point of view can also be found in~\cite{Parsegov}.

The paper is organized as follows: in Section~\ref{sec:prel} some preliminary results
are given. In Section~\ref{sec:main} the main result is
provided. Finally, in Section~\ref{sec:num}, two numerical examples are provided to illustrate
how the FJ model converges at steady state to a quasi-consensus condition. Eventually, in Section VI
some concluding remarks are given.

%
\section{Notation and Preliminaries} \label{sec:prel}
Given a square matrix $A$, $\rho(A)$, $\det(A)$, $\bar\sigma(A)=\|A\|$, $\ubar\sigma(A)$, and $\lambda_i(A)$ will denote, respectively, the spectral radius, the determinant, the maximum singular value, the minimum singular value, and the $i$-th eigenvalue of $A$. $\bigO(\cdot)$ denotes the Big O Laundau symbols:
\begin{align*}
A(\epsilon)=\bigO(\epsilon)\Leftrightarrow & \,\exists\,\bar\epsilon>0\,\, \textrm{and}\,\, \exists\, M>0: \\
&\,\,\|A(\epsilon)\|\leq M\epsilon\,\, \forall\epsilon \in ]0,\bar\epsilon]\,.
\end{align*}
Finally, $\mathds{1}$ will denote the column vector with all elements equal to $1$, its dimension will be clear by the context.

Given an invertible matrix $A$ and a matrix $X$ we will make use of the following matrix equality
\begin{equation}
(A+\epsilon X)^{-1}=A^{-1}-\epsilon A^{-1}XA^{-1}+\bigO(\epsilon^2) \label{eq:expans}\,.
\end{equation}
Moreover given an invertible matrix $A$ and the vectors $u$ and $v$, the matrix $A+uv^T$ is invertible \emph{if and only if} $1+v^TA^{-1}u\neq 0$, and
\begin{equation}
(A+uv^T)^{-1}=A^{-1}-\frac{A^{-1}uv^TA^{-1}}{1+v^TA^{-1}u}\,. \label{eq:woodbury}
\end{equation}
Equation~\eqref{eq:expans} is obtained by truncating to the first order, the Taylor series of $(A+\epsilon X)^{-1}$ with respect to the scalar $\epsilon$. Whereas equation~\eqref{eq:woodbury} is the Sherman-Morrison formula~\cite{Sherman}.

Given the square matrices $A$ and $X$, then the following relationship holds locally around $\epsilon=0$ (see for example~\cite{Kato})

\begin{equation} \label{eq:eigBound1}
|\lambda_i[A+\epsilon X]-\lambda_i[A]|\leq k_{A}|\epsilon|^{1/m_i}\,,
\end{equation}
where $k_{A}$ is a positive constant, and $m_i$ is the multiplicity of $\lambda_i[A]$; if $\lambda_i[A]$ is a simple real eigenvalue then
\begin{equation} \label{eq:eigBound2}
\lambda_i[A+\epsilon X]-\lambda_i[A]=\epsilon \frac{y^T X x}{y^Tx}+\bigO(\epsilon^2)\,.
\end{equation}
where $y^T$ and $x$ are left and right eigenvectors of $A$ corresponding with the eigenvalue $\lambda_i[A]$. 

Given a matrix $A\in\mathbb{R}^{n\times n}$ with all row vectors having Euclidean norm less than a positive scalar $\beta$, then the following inequality holds
\begin{equation} \label{eq:sigmamin}
\ubar\sigma (A)\geq\left(\frac{n-1}{n\beta^2}\right)^\frac{n-1}{2}|\det(A)|\,.
\end{equation}
Inequality~\eqref{eq:sigmamin} can be immediately derived from the main result of~\cite{Hong}.

A square matrix $A$ with nonnegative elements is called row stochastic if and only if the sum of all the elements on a row is equal to $1$, irreducible if and only if the directed graph $\mathcal{G}(A)$ is strongly connected, and primitive if and only $A$ has only one eigenvalue $r$ for which $|r|=\rho(A)$ (see~\cite{Meyer}, Chap. 8, for an extensive discussion about the Perron-Frobenius theory of nonnegative matrices). 

We will use the FJ model~\eqref{eq:Friedkin1} in the slightly different form
\begin{equation} \label{eq:Friedkin2}
y(k+1)=(I-\Sigma) W y(k)+\Sigma y_0\,,\quad y(0)=y_0\,,
\end{equation}
where $\Sigma=\textrm{diag}(\sigma_1,\sigma_2,\dots,\sigma_n)=I-\Lambda$; the scalars $\sigma_i$ can be interpreted as coefficients weighting the \emph{immunity of an agent to interpersonal influence}. 

In the rest of the paper, we will make use of the following assumptions.
\begin{assumption}\label{ass:irred}
$W$ is a nonnegative row stochastic irreducible matrix.
\end{assumption}
Since $W$ is row stochastic it is simple to recognize that the vector $\mathds{1}$ is a right eigenvector with $1$ as eigenvalue. Moreover the Perron-Frobenius theory allows to conclude that $1$ is the simple eigenvalue corresponding to the Perron root of $W$, and that it is possible to choose a corresponding left eigenvector $\alpha^T$ with nonnegative components~\cite{Meyer}. In what follows we will assume that $\alpha^T$ is normalized in such a way to have $\alpha^T\mathds{1}=1$.

The following lemma states a well know fact about the convergence to a consensus value for the DeGroot model~\eqref{eq:DeGroot} (see~\cite{DeGroot}).
\begin{lemma} \label{thm:DeGroot}
If Assumption~\ref{ass:irred} holds and $W$ is a primitive matrix, then the solutions of~\eqref{eq:DeGroot} satisfy
\begin{equation*}
\lim_{k\rightarrow\infty} y(k)=\bar y = \mathds{1}\alpha^T y_0\,.
\end{equation*}
\end{lemma}
\begin{remark}
Lemma~\ref{thm:DeGroot} shows that under the additional assumption of $W$ being a primitive matrix, each agent's opinion converge to the value $\alpha^Ty_0$, i.e. the final opinion of each agent is a convex combination, according to the coefficients $\alpha_i$, of the initial opinions. The assumption of $W$ being primitive is necessary to exclude the existence of other eigenvalues of $W$ located on the unit circle, which would prevent convergence for all initial conditions.
\end{remark}

\begin{assumption}\label{ass:Sigma} $\Sigma=\emph{diag}(\sigma_1,\sigma_2,\dots,\sigma_n)$, where $\sigma_i\in]0,1-\epsilon]$, with $\epsilon>0$ small. We denote with $\mathcal{S}(\epsilon)$ the set of admissible value for $\Sigma$.  
\end{assumption}

Under Assumption~\ref{ass:Sigma} we have that $\sigma_M=\max_i \sigma_i$ is always positive and less than 1. We can define the coefficients $p_i=\sigma_i/\sigma_M$ and let
\begin{equation}
\Sigma=\sigma_M\tilde{\Sigma}\,,
\end{equation}
where $\tilde{\Sigma}=\textrm{diag}(p_1,p_2,\dots,p_n)$, $\bar\sigma(\tilde\Sigma)=1$, and $\ubar \sigma(\tilde\Sigma)=\min_i p_i>0$.

The behavior of the FJ model~\eqref{eq:Friedkin2} depends on the matrix $\Sigma$, and hence on the parameters $\sigma_M$ and $\tilde \Sigma$, in what follows we will consider the matrix $\tilde \Sigma$ fixed, and we will study the behavior of the model for small values of $\sigma_M$, for this reason in what follows we will omit dependences on $\tilde \Sigma$.

Let $w_i^T$ be the $i$-th row of the matrix W, then the dynamic matrix of the FJ model~\eqref{eq:Friedkin1} is given by
\begin{equation} \label{eq:Wt}
\tilde W(\sigma_M)=(I-\Sigma)W=\begin{pmatrix}(1-\sigma_1)w_1^T\\(1-\sigma_2)w_2^T\\ \vdots \\ (1-\sigma_n) w_n^T\end{pmatrix}\,,
\end{equation}

Now we are ready to state the first result of the paper.
\begin{lemma} \label{thm:stab}
If Assumptions~\ref{ass:irred} and~\ref{ass:Sigma} holds then the dynamic matrix of the FJ model is a nonnegative irreducible matrix, with all eigenvalues inside the open unit circle. Moreover its spectral radius is a strictly decreasing function of $\sigma_M$. 
\end{lemma}
\begin{proof}
Since $\sigma_i\in]0, 1[$ it is clear that all the elements of $\tilde W$ are nonnegative, moreover an element of $\tilde W$ is $0$ if and only if the corresponding element in $W$ is $0$. It follows that the irreducibility of $W$ implies the irreducibility of $\tilde W$. This proves the fact that $\tilde W$ is a nonnegative irreducible matrix.

Now, since $W$ is row stochastic and irreducible we known that its spectral radius is $1$, moreover from~\eqref{eq:Wt} it is simple to recognize that, given $\sigma_{M1}>\sigma_{M2}$, the following entrywise inequalities are satisfied  
\begin{equation} \label{eq:wie}
\tilde{W}(\sigma_{M1})\leq \tilde{W}(\sigma_{M2})\leq W\,,\quad \forall\, \Sigma\in\mathcal{S}(\epsilon)\,,
\end{equation}
where the inequalities are strict for any nonzero element of $W$.  
It follows from the Wielandt's theorem (see Theorem 8.3.11 of~\cite{Meyer}) that
\begin{equation*}
\rho(\tilde W(\sigma_{M1}))<\rho(\tilde W(\sigma_{M2}))<\rho(W)=1\,.
\end{equation*}
This completes the proof.
\end{proof}
 \begin{remark}
Lemma~\ref{thm:stab} shows that if Assumptions~\ref{ass:irred} and~\ref{ass:Sigma} hold, the FJ model is asymptotically stable, so as the steady state value
\begin{equation}
\bar y= \left[I-(I-\Sigma)W\right]^{-1}\Sigma y_0
\end{equation}
is well defined for all $\Sigma\in\mathcal{S}(\epsilon)$.  More general discussions about the convergence of the FJ model can be found in~\cite{Parsegov} and~\cite{Prosku},
\end{remark}
In the next section we will deal with the problem of characterizing the asymptotic behavior of the static gain matrix
\begin{equation}\label{eq:Hsigma}
H(\sigma_M)=\sigma_M\left[I-(I-\sigma_M\tilde\Sigma)W\right]^{-1}\tilde\Sigma\,,
\end{equation}
as $\sigma_M$ approaches $0$ from the right. Note that, although $H(\sigma_M)$ is well defined for $\sigma_M\in]0, 1[$, since the matrix $I-W$ is not invertible (having $W$ an eigenvalue in $1$), its value for $\sigma_M=0$ cannot be directly evaluated using equation~\eqref{eq:Hsigma}.
\begin{remark}
The statement of Lemma~\ref{thm:stab} holds also if we allow $\sigma_i=0$ for some, but not all, indices $i$. This can be easily proved using the Wiedlandt's theorem, and considering that inequalities~\eqref{eq:wie} continue to hold strictly on at least some elements on a row of $W$ and $\tilde W$.
\end{remark}

We close this Section with a result on a spectral property of the matrix $H(\sigma_M)$.
\begin{lemma} \label{lemma:eigen}
If Assumptions~\ref{ass:irred} and~\ref{ass:Sigma} hold then $H(\sigma_M)$ has an eigenvalue in $1$, moreover $\mathds{1}$ and $\alpha^T\Sigma(I-\Sigma)^{-1}/(\alpha^T\Sigma(I-\Sigma)^{-1}\mathds{1})$ are, respectively, two corresponding right and left eigenvectors.
\end{lemma}
\begin{proof}
The inverse of $H(\sigma_M)$ is given by
\begin{equation}
H^{-1}(\sigma_M)=\Sigma^{-1}(I-W+\Sigma W)\,,
\end{equation}
It is simple to verify that
\begin{align}
H^{-1}(\sigma_M)\mathds{1}&=\mathds{1} \\
\frac{\alpha^T\Sigma(I-\Sigma)^{-1}}{\alpha^T\Sigma(I-\Sigma)^{-1}\mathds{1}}H^{-1}(\sigma_M)&=\frac{\alpha^T\Sigma(I-\Sigma)^{-1}}{\alpha^T\Sigma(I-\Sigma)^{-1}\mathds{1}}\,.
\end{align}
from which the proof follows. 
\end{proof}
\begin{remark}
The left eigenvector, defined in the statement of Lemma~\ref{lemma:eigen}, has been normalized in such a way that the sum of its components (which are non negative) is $1$.
\end{remark}

\section{Main results} \label{sec:main}
In this section we will prove our main result, i.e. that for sufficiently small $\sigma_M>0$ the FJ model~\eqref{eq:Friedkin2} asymptotically converge to a quasi-consensus value. 
We first prove that $H(\sigma_M)$ is a bounded matrix.

\begin{lemma} \label{lemma:bound}
If Assumptions~\ref{ass:irred} and~\ref{ass:Sigma} hold then $\|H(\sigma_M)\|$ is bounded.
\end{lemma}
\begin{proof}
Since in each interval $\sigma_M\in[\delta,1-\delta]$, with $0<\delta<0.5$, the matrix function $\Sigma^{-1}(I-W+\Sigma W)$ is continuous and invertible with respect to $\sigma_M$, it is clear that $H(\sigma_M)$ is bounded when $\sigma_M>\delta$. It remains to prove that $H(\sigma_M)$ remains bounded when $\sigma_M\rightarrow 0^+$. 

We have
\begin{align}
\|H(\sigma_M)\|&= \bar\sigma\left[\sigma_M(I-(I-\sigma_M\tilde\Sigma)W)^{-1}\tilde \Sigma\right] \nonumber\\
&\leq\frac{\sigma_M\bar\sigma(\tilde\Sigma)}
{\ubar{\sigma}\left[I-(I-\sigma_M\tilde\Sigma)W\right]}\nonumber \\
&=\frac{\sigma_M}{\ubar{\sigma}\left[I-(I-\sigma_M\tilde\Sigma)W\right]}\,. \label{eq:Bound1}
\end{align}
We now proceed to find a lower bound for the minimum singular value appearing in equation~\eqref{eq:Bound1}, when $\sigma_M$ becomes arbitrarily small.

First of all, it is easy by direct inspection to show that each row of the matrix $I-(I-\Sigma)W$ has Euclidean norm bounded by $2$. So as, using inequality~\eqref{eq:sigmamin}, we obtain
\begin{align} \label{eq:sigmaminlow}
&\ubar{\sigma}\left[(I-(I-\Sigma)W)\right]\geq \nonumber\\ 
&\left(\frac{n-1}{4n}\right)^\frac{n-1}{2}\prod_{i=1}^n|\lambda_i(I-(I-\Sigma)W)|\,.
\end{align}

Now, consider the the eigenvalues of $(I-\Sigma)W$. They depends continuously on $\sigma_M$, so as we can associate to each of them an eigenvalue of $W$. Assuming that the eigenvalues of $W$ are ordered in increasing value of the modulus, in such a way that the Perron-Frobenius eigenvalue is the last one, and denoting with $m_i$ the multiplicity of the $i$-th eigenvalue of $W$, inequality~\eqref{eq:eigBound1} allows to write
\begin{equation}\label{eq:lbound1}
|\lambda_i[(I-\Sigma)W]-\lambda_i[W]|<k_W\sigma_M^{1/m_i}\,,
\end{equation}
where the positive constant $k_W$ depends only on $W$. From which
\begin{equation} \label{eq:eigen1}
|1-\lambda_i[(I-\Sigma)W]|=|1-\lambda_i[W]|+\bigO(\sigma_M^{1/m_i})\,,
\end{equation}
for $i=1,2,\dots,n-1$.

Using equality~\eqref{eq:eigBound2}, we also obtain
\begin{align}
1-\lambda_n[(I-\Sigma)W]&=\frac{\alpha^T\tilde\Sigma W\mathds{1}}{\alpha^T\mathds{1}}\sigma_M+\bigO(\sigma_M^2) \nonumber \\
&=\sigma_M\alpha^T\tilde\Sigma\mathds{1}+\bigO(\sigma_M^2)\,.\label{eq:eigen2}
\end{align}

Now define
\begin{align*}
d_i&=|1-\lambda_i[W]|\,,\quad i=1,2,\dots,n-1 \\
d_0&=\min_{i=1,2,\dots,n-1} \,d_i \\
m_0&=\max_{i=1,2,\dots,n-1}\, m_i\,.
\end{align*}
Since the first $n-1$ eigenvalues of $W$ are distinct from $1$, $d_0$ is a positive constant depending only on $W$.

Putting together equations from~\eqref{eq:sigmaminlow} to~\eqref{eq:eigen2}, we obtain  
\begin{align}
&\ubar{\sigma}\left[I-(I-\Sigma)W\right]\geq \nonumber\\
&\left(\frac{n-1}{4n}\right)^\frac{n-1}{2}\sigma_M d_0^{n-1}\alpha_T\tilde\Sigma\mathds{1} +\bigO\left(\sigma_M^{(m_0+1)/m_0}\right)\,. \label{eq:Bound2}
\end{align}

From which, letting
\begin{equation*}
c_W=\left(\frac{n-1}{4n}\right)^\frac{n-1}{2}d_0^{n-1}\,,
\end{equation*}
it is possible to write
\begin{align*}
\|H(\sigma_M)\|\leq&\frac{\sigma_M}{c_W\sigma_M\alpha^T\tilde\Sigma\mathds{1}+\bigO(\sigma_M^{(m_0+1)/m_0})}\\
&\mathop{\longrightarrow}^{\sigma_M\rightarrow 0^+} \frac{1}{c_W\alpha^T\tilde\Sigma\mathds{1}}\,,
\end{align*}
which shows the boundedness of $\|H(\sigma_M)\|$ also for arbitrarily small values of $\sigma_M$.
\end{proof}

Next step is to show that for $\sigma_M$ sufficiently close to zero the matrix $H$ can be approximated by a given rank one matrix. In order to introduce in a simple way this approximation, we first derive it by means of not fully rigorous arguments.   

Consider equation~\eqref{eq:Hsigma}, from which we obtain
\begin{equation*}
\left[I-(I-\sigma_M\tilde\Sigma)W\right]H(\sigma_M) =\sigma_M\tilde\Sigma\,,
\end{equation*}
and
\begin{equation*}
H(\sigma_M)-WH(\sigma_M)+\sigma_M[\tilde\Sigma WH(\sigma_M)-\tilde\Sigma]=0\,.
\end{equation*}
Hence, being $H(\sigma_M)$ a bounded matrix, we can write
\begin{equation} \label{eq:HsigmaEq}
H(\sigma_M)-WH(\sigma_M)=\bigO(\sigma_M)\,.
\end{equation}

Now, let us assume that there exists a matrix $\bar H$ such that
\begin{equation}
\bar H = \lim_{\sigma_M\rightarrow0^+}H(\sigma_M)\,.
\end{equation}
Then $\bar H$ has to satisfy the linear equation
\begin{equation}\label{eq:barHeq}
\bar H = W\bar H\,.
\end{equation}
It follows that the columns of $\bar H$ have to be right eigenvectors of the eigenvalue in $1$ of $W$, i.e. $\bar H$ has to be in the form
\begin{equation*}
\bar H=\begin{pmatrix} l_1\mathds{1} &l_2\mathds{1}& \dots & l_n\mathds{1}\end{pmatrix}=\mathds{1}l^T\,,
\end{equation*}
where $l_i$, $i=1,2,\dots,n$ are suitable scalars, and $l^T=\begin{pmatrix}
l_1 & l_2 &\dots & l_n\end{pmatrix}$.

On the other hand we already know that for every $0<\sigma_M<1$, $H(\sigma_M)$ has $\mathds{1}$ as a right eigenvector, and ${\alpha^T\Sigma(I-\Sigma)^{-1}}/{(\alpha^T\Sigma(I-\Sigma)^{-1}\mathds{1})}\approx {\alpha^T\tilde\Sigma}/{(\alpha^T\tilde\Sigma\mathds{1})} $ as a left eigenvector (Lemma~\ref{lemma:eigen}), so as we can expect that the same holds also for $\bar H$. As a consequence the following equalities are also expected to hold
\begin{align*}
&\bar H\mathds{1}=\mathds{1}l^T\mathds{1}=\mathds{1}\Rightarrow\sum_{i=1}^n l_i=1\,. \\
&\frac{\alpha^T\tilde\Sigma}{\alpha^T\tilde\Sigma\mathds{1}}\mathds{1}l^T=\frac{\alpha^T\tilde\Sigma}{\alpha^T\tilde\Sigma\mathds{1}}\,,
\end{align*}
from which
\begin{equation*}
l^T=\frac{\alpha^T\tilde\Sigma}{\alpha^T\tilde\Sigma\mathds{1}}\,.
\end{equation*}
In conclusions if $H(\sigma_M)$ can be approximated by a constant matrix around $\sigma_M=0$, we expect that this matrix is
\begin{equation}\label{eq:barH}
\bar H = \mathds{1}\frac{\alpha^T\tilde\Sigma}{\alpha^T\tilde\Sigma\mathds{1}}\,.
\end{equation}

Next theorem which is the main result of the paper show that this is, indeed, true.
\begin{theorem} \label{thm:main}
If Assumptions~\ref{ass:irred} and~\ref{ass:Sigma} hold then 
\begin{equation}
H(\sigma_M)=\mathds{1}\frac{\alpha^T\tilde\Sigma}{\alpha^T\tilde\Sigma\mathds{1}}+\bigO(\sigma_M)
\end{equation}
\end{theorem}
\begin{proof}
For the sake of conciseness let us introduce the following quantities
\begin{align*}
q_1&=\alpha^T\tilde\Sigma\mathds{1}>0\\
q_2&=\alpha^T\tilde\Sigma^2\mathds{1}>0\\
h&=\alpha^T\tilde\Sigma(I-\Sigma)^{-1}\mathds{1}>0\,.
\end{align*}
Now consider the following matrix
\begin{equation*}
\tilde H=\mathds{1}\frac{\alpha^T\tilde\Sigma (I-\Sigma)^{-1}}{\alpha^T\tilde\Sigma\mathds{1}}=\mathds{1}\frac{\alpha^T\tilde\Sigma (I-\Sigma)^{-1}}{q_1}\,.
\end{equation*}

First we prove that $H(\sigma_M)-\tilde H$ is invertible for sufficiently small values of $\sigma_M>0$. Indeed, $\tilde H$ is a rank one matrix, and for sufficiently small positive values of $\sigma_M$ we have
\begin{align*}
\alpha^T\tilde\Sigma(I-\Sigma)^{-1}H(\sigma_M)^{-1}\mathds{1}&=h\\
&=\sum_{i=1}^n\frac{\alpha_ip_i}{1-\sigma_i}\\
&>\sum_{i=1}^n{\alpha_ip_i}=q_1\,,
\end{align*}
from which
\begin{equation*}
1-\displaystyle\frac{\alpha^T\tilde\Sigma(I-\Sigma)^{-1}H(\sigma_M)^{-1}\mathds{1}}{q_1}=1-\frac{h}{q_1}<0\,.
\end{equation*}
It follows that for positive sufficiently small values of $\sigma_M$, it is possible to apply the Sherman-Morrison formula~\eqref{eq:woodbury}, obtaining
\begin{align*}
(H(\sigma_M)-\tilde H)^{-1}&=H^{-1}(\sigma_M)+\\
&\frac{H^{-1}(\sigma_M)\mathds{1}\alpha^T\tilde\Sigma(I-\Sigma)^{-1}H^{-1}(\sigma_M)}{q_1(1-\displaystyle\frac{\alpha^T\tilde\Sigma(I-\Sigma)^{-1}H(\sigma_M)^{-1}\mathds{1}}{q_1})}\\
&=H^{-1}(\sigma_M)+\mathds{1}\alpha^T\tilde\Sigma\displaystyle\frac{I-\sigma_M\tilde\Sigma+\bigO(\sigma_M^2)}{\sigma_Mq_2+\bigO(\sigma_M^2)}\,.
\end{align*}
Moreover
\begin{align*}
H^{-1}(\sigma_M)=\frac{1}{\sigma_M}(\tilde\Sigma^{-1}(I-W))+ W\,,
\end{align*}
from which
\begin{equation*}
(H(\sigma_M)-\tilde H)^{-1}=\frac{1}{\sigma_M}\left( A+\sigma_M B\right)+\bigO(\sigma_M)\,,
\end{equation*}
where
\begin{align*}
A&=\tilde\Sigma^{-1}(I-W)+\frac{1}{q_2}\mathds{1}\alpha^T\tilde\Sigma\\
B&= W- \frac{1}{q_2}\mathds{1}\alpha^T\tilde\Sigma^2\,.
\end{align*}
Now, $A$ is an invertible matrix, this can be explained with the fact that $\tilde\Sigma^{-1}(I-W)$ has a single eigenvalue in $0$ to whom corresponds as a right and a left eigenvectors, respectively, $\mathds{1}$, and $\alpha^T\tilde\Sigma$. It follows that the eigenvalues of $A$ are $q_1/q_2$ and the nonzero eigenvalues of $\tilde\Sigma^{-1}(I-W)$.

Using~\eqref{eq:expans}, we obtain
\begin{align*}
H(\sigma_M)-\tilde H&=\left[(H(\sigma_M)-\tilde H)^{-1}\right]^{-1} \\
&=\sigma_M\left(A^{-1}-\sigma_MA^{-1}BA^{-1}+\bigO(\sigma_M)\right)\,,
\end{align*}
from which
\begin{equation*}
H(\sigma)=\tilde{H}+\bigO(\sigma)\,.
\end{equation*}
Now it is simple to prove that
\begin{align*}
\tilde H- \bar H&=\frac{\mathds{1}}{q_1}\alpha^T\tilde\Sigma\left(I-(I-\Sigma)^{-1}\right)
\\&=\frac{\mathds{1}}{q_1}\alpha^T\tilde\Sigma\left(\sigma_M\tilde\Sigma+\bigO(\sigma_M^2))\right)=\bigO(\sigma_M)
\end{align*}
So as we can also write
\begin{equation*}
H(\sigma)=\bar H + \bigO(\sigma_M)\,,
\end{equation*}
which concludes the proof.

Note that the initial use of the matrix $\tilde H$ instead of $\bar H$ is due to avoid a singularity in equality~\eqref{eq:woodbury}
\end{proof}

\begin{corollary} \label{cor:main}
If Assumptions~\ref{ass:irred} and~\ref{ass:Sigma} hold, the FJ model~\eqref{eq:Friedkin2} steady state value is given by
\begin{equation*}
\bar y= \lim_{k\rightarrow \infty} y(k)=\bar H y_0 + \bigO(\sigma_M)\,.
\end{equation*}
As a consequence, given two components $i$ and $j$ of $\bar y$, the following quasi-consensus condition is also verified
\begin{align*}
\exists\, \bar\sigma_M>0\,\, \emph{and}\,\,&\exists\,M>0:\\
&\,\,|\bar y_i -\bar y_j|<M\sigma_M \,,\,\forall\sigma_M\in]0,\bar\sigma_M]\,.
\end{align*}
\end{corollary}
\begin{proof} It is an immediate consequence of Theorem~\ref{thm:main}.
\end{proof}
\begin{remark}
Theorem~\ref{thm:main} was already proved in~\cite{Prosku}  (see Lemma~24) for the case where $\tilde\Sigma=I$. 
\end{remark}

Denoting with $y(\sigma_M,k)$ the solutions of equation~\eqref{eq:Friedkin2}, and using Corollary~\ref{cor:main}, we can write
\begin{equation} \label{eq:lim}
\lim_{\sigma_M\rightarrow 0^+} \,\lim_{k\rightarrow \infty} y(\sigma_M,k)=\bar Hy_0=\frac{\mathds{1}\alpha^T\tilde\Sigma}{\alpha^T\tilde\Sigma\mathds{1}}y_0\,,
\end{equation}
whereas, we know that for any finite $k$ we have
\begin{equation} \label{eq:lim1}
\lim_{\sigma_M\rightarrow 0^+} y(\sigma_M,k)=y(0,k)=y_G(k)=W^k y_0\,.
\end{equation}
where $y_G(k)$ is the solution of the DeGroot model~\eqref{eq:DeGroot}.
If $W$ is primitive then
\begin{equation} \label{eq:lim2}
\lim_{k\rightarrow \infty}\, \lim_{\sigma_M\rightarrow 0^+}  y(\sigma_M,k)=\mathds{1}\alpha^Ty_0\,.
\end{equation} 
It follows that the two limits in~\eqref{eq:lim} cannot in general be inverted. On the other hand for sufficiently small value of $\sigma_M$ and for finite value of $k$ the solutions of the FJ and DeGroot models will be close.

\begin{remark} \label{rem:timescale}
In the proof of Theorem~\ref{thm:main} we established that for sufficiently small value of $\sigma_M$ the eigenvalues of the dynamic matrix $\tilde W$ are very close to the eigenvalues of $W$. If $W$ is primitive and $\sigma_M$ is sufficiently small we can therefore conclude that the evolution of the solutions of the FJ model develop on two time scales. In the first part of the evolution the state evolve following the dynamics of the DeGroot model, reaching eventually an approximate first opinions consensus on the value given by $\mathds{1}\alpha^Ty_0$, then on a longer time scale (dictated by the Perron-Froebenius eigenvalue of $\tilde W$, the agent opinions shifts toward the final value $\mathds{1}\alpha^T\tilde\Sigma y_0/(\alpha^T\tilde\Sigma\mathds{1})$.

In other words it can happen that the agent opinions first reach a consensus based on the DeGroot model, then the consensus slowly shifts on a different value depending on the matrix $\tilde\Sigma$.   
\end{remark}

\section{Numerical example} \label{sec:num}
In this section we will present two examples; since their scope is only to show the results of this paper from a numerical point of view, these examples are not linked to any particular applications.

Consider the FJ model~\eqref{eq:Friedkin2}, with the following parameters
\begin{subequations} \label{eq:ex1}
\begin{align}
W&=\begin{pmatrix}
0.8 &   0.10 &    0.05  &  0.05 \\
    0.30   &  0.40  &   0.20  &   0.10 \\
    0.10   &  0.10   &  0.60  &  0.20 \\
    0.10  &  0.30  &   0.30  &   0.30
\end{pmatrix}\\
\tilde\Sigma&=\textrm{diag}
\begin{pmatrix}
0.5 & 1 & 0.2 & 0.1\end{pmatrix}\,.
\end{align}
\end{subequations}
clearly, for $\sigma_M\in]0,1[$, the matrices $W$ and $\Sigma=\sigma_M\tilde \Sigma$  satisfy Assumptions~\ref{ass:irred} and~\ref{ass:Sigma}. 

Table~\ref{tab:tab1} quantifies the maximum distance between the elements of the steady state response for different values of $\sigma_M$, when $y_0=\begin{pmatrix}
0.20  & 0.50 & 0.01 & 0.29 \end{pmatrix}^T$. The last column shows the $95\%$ settling time. Clearly, as $\sigma_M$ approaches $0$, the spectral radius of $\tilde W$ approaches $1$, and the speed of convergence increases approximately with the reciprocal of $\sigma_M$. Figures~\ref{fig:Fig1} and~\ref{fig:Fig1a} show, respectively, the time behavior of $y$ when $\sigma_M=0.05$, and when $\sigma_M=0.01$. 

In these figures, according to Remark~\ref{rem:timescale}, it is possible to notice  two timescales along which the behavior of the FJ model solution develops; first all the elements of $y$ move towards the consensus value of the DeGroot model ($0.22$ in this case), then they slowly converge to the steady state values (close to the quasi-consensus value of $0.30$). Figure~\ref{fig:comparison} show a comparison between the FJ and DeGroot model solutions for the first time steps when $\sigma_M=0.01$.

\renewcommand{\arraystretch}{1.2}
\begin{table} \label{tab:tab1}
\begin{center}
\caption{Numerical results for the case where the FJ model matrices are defined in equation~\eqref{eq:ex1}.}
\begin{tabular}{ |c | c | c | c| }
    \hline
    $\sigma_M$ & $\|H-\bar H\|/\|H\|$ & $\max |\bar y_i-\bar y_j|$ & $T_a$ \\
    \hline
    $0.2$ & $2.3\times 10^{-1}$ & $7.0\times 10^{-2}$ & $22$ \\ \hline
    $0.05$ & $6.2\times 10^{-2}$ & $1.9\times 10^{-2}$ & 76 \\ \hline
    $0.01$ & $1.3\times 10^{-2}$ & $3.8\times 10^{-3}$ & 361  \\ \hline
    $0.001$ &$1.3\times 10^{-3}$  & $3.8\times 10^{-4}$& 3572 \\
    \hline
  \end{tabular}
\end{center}
\end{table}

\begin{figure}[thb]
\begin{center}
\includegraphics[width=7.5cm]{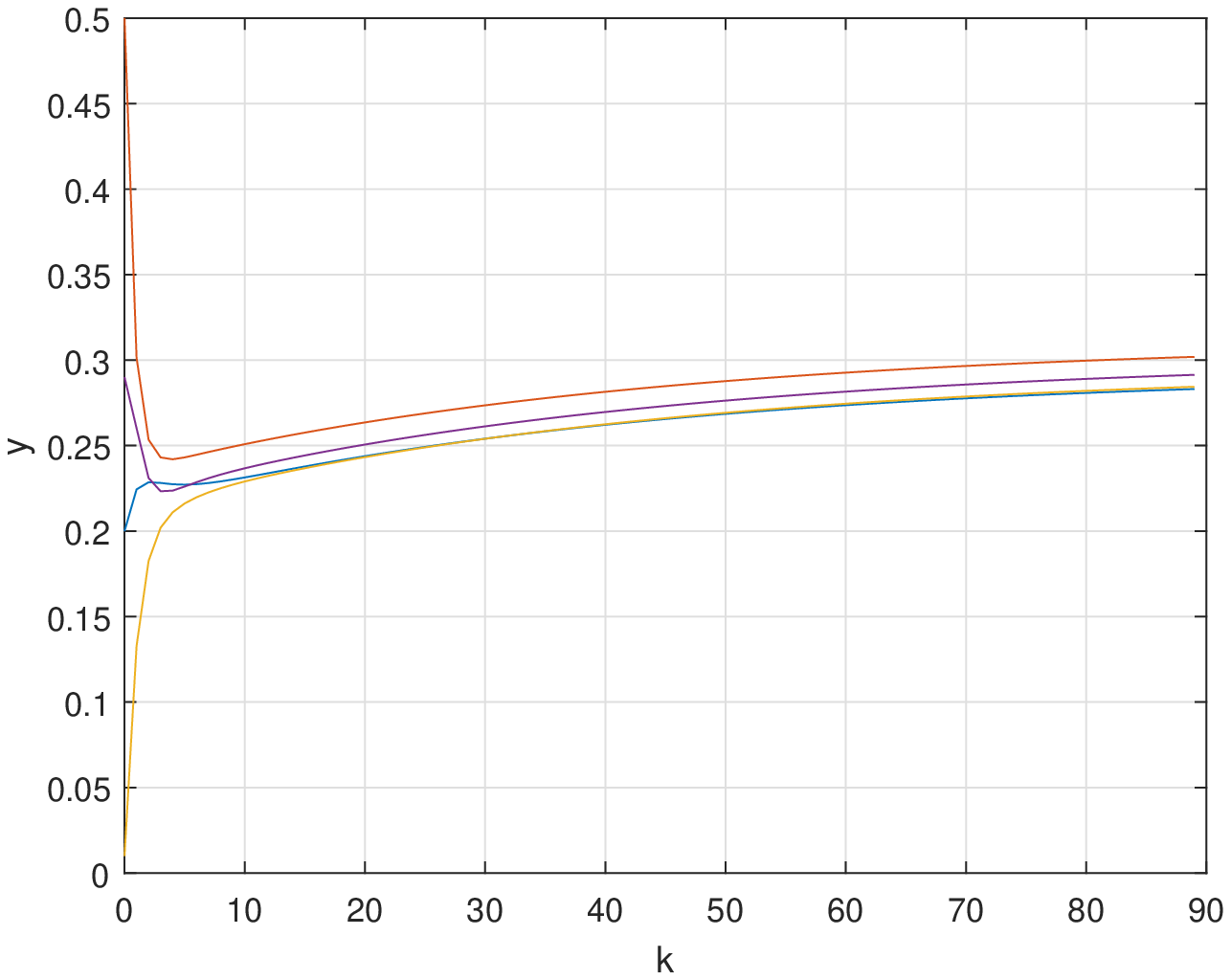}
\caption{Behavior of the FJ model with the matrices specified in equation~\eqref{eq:ex1} when $\sigma_M=0.05$.}
\label{fig:Fig1}
\end{center}
\end{figure}

\begin{figure}[thb]
\begin{center}
\includegraphics[width=7.5cm]{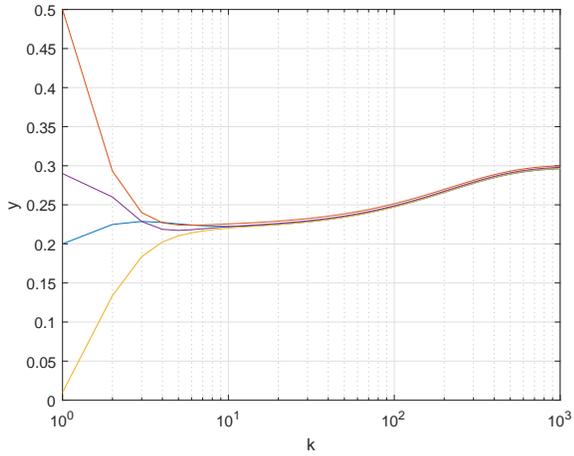}
\caption{Behavior of the FJ model with the matrices specified in equation~\eqref{eq:ex1} when $\sigma_M=0.01$, a logarithmic scale is used for the x-axis.}
\label{fig:Fig1a}
\end{center}
\end{figure}

\begin{figure}[thb]
\begin{center}
\includegraphics[width=7.5cm]{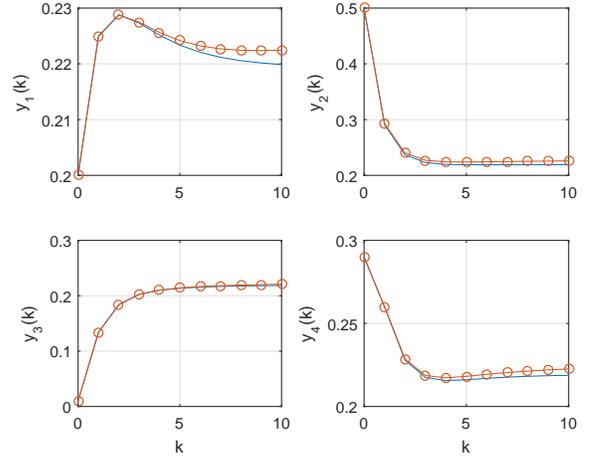}
\caption{Comparison between the solutions of the FJ (solid line with circles) and DeGroot (solid lines) models for $k=0,2,\dots,10$ ($\sigma_M=0.01$).} 
\label{fig:comparison}
\end{center}
\end{figure}
To complete the analysis, we also considered a case where the matrix $W$ is not primitive. Since $W$, apart from $1$, has other eigenvalues on the unit circle, the corresponding DeGroot model does not converge for all initial conditions. 
\begin{subequations} \label{eq:ex2}
\begin{align}
W&=\begin{pmatrix}
0 & 1 & 0 & 0 \\
    \frac{2}{3} & 0 & \frac{1}{3} &0 \\
    0 &\frac{1}{3} &0& \frac{2}{3} \\
    0 &0& 1& 0
\end{pmatrix}\\
\tilde\Sigma&=\textrm{diag}
\begin{pmatrix}
0.5 & 1 & 0.2 & 0.1\end{pmatrix}
\end{align}
\end{subequations}

Figure~\ref{fig:fignonprim} shows the time behavior of $y$ when $W$ and $\tilde \Sigma$ are defined as in equation~\eqref{eq:ex2}, and $\sigma_M=0.1$. In this case the eigenvalues of $W$ are in $\{-1, -0.68, 0.68, 1\}$, and the quasi-consensus value is $0.37$.

\begin{figure}[thb]
\begin{center}
\includegraphics[width=7.5cm]{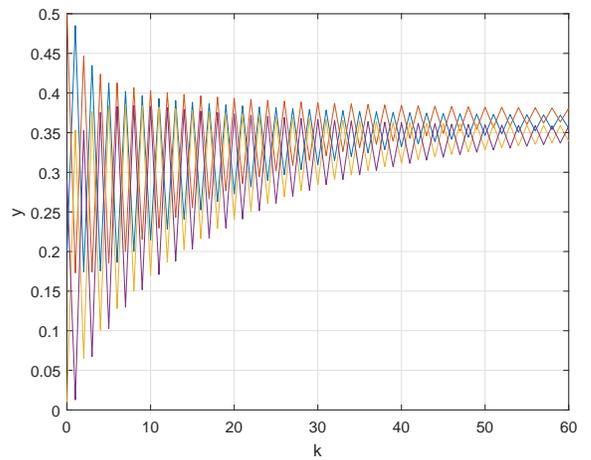}
\caption{Behavior of the FJ model with the matrices specified in equation~\eqref{eq:ex2} when $\sigma_M=0.1$}
\label{fig:fignonprim}
\end{center}
\end{figure}

\section{Conclusion} \label{sec:concl}
In this paper we made a detailed analysis of the asymptotic behavior of the Friedkin-Johnsen model of opinion dynamics when the coefficients weighting the agent susceptibilities to interpersonal influence approach $1$ along a general direction. We show that under suitable assumptions, if these weights are sufficiently close to $1$, then the agent opinions converge to very close values, in other words a quasi-consensus condition is achieved.


\end{document}